 \UseRawInputEncoding
 \documentclass[11pt, reqno]{article}
\usepackage[leqno]{amsmath}
\usepackage{amssymb,amsthm,graphicx,enumitem, mathtools}

\theoremstyle{definition}
\newtheorem{dfn}{Definition}[section]
\newtheorem{dfn*}{Definition}
\newtheorem{exam}[dfn]{Example}
\newtheorem{remark}[dfn]{Remark}

\theoremstyle{plain}
\newtheorem{prop}[dfn]{Proposition}
\newtheorem{prop*}[dfn*]{Proposition}

\usepackage{hyperref}
\usepackage{doi}
\usepackage{enumitem}
\usepackage{mathtools}

\usepackage{float}

\begin{document}

\title{On the Oval Shapes of Beach Stones}

\author{Theodore P. Hill}

\date{\vspace{-5ex}}  

\maketitle

\begin{abstract}
This article introduces a new geophysical theory, in the form of a single simple partial integro-differential equation, to explain how frictional abrasion alone of a stone on a planar beach can lead to the oval shapes observed empirically. The underlying idea in this theory is the intuitive observation that the rate of ablation at a point on the surface of the stone is proportional to the product of the curvature of the stone at that point and how often the stone is likely to be in contact with the beach at that point. Specifically, key roles in this new model are played by both the random wave process and the global (non-local) shape of the stone, i.e., its shape away from the point of contact with the beach. 
The underlying physical mechanism for this process is the conversion of energy from the wave process into potential energy of the stone. No closed-form or even asymptotic solution is known for the basic equation, even in a 2-dimensional setting, but basic numerical solutions are presented in both the deterministic continuous-time setting using standard curve-shortening algorithms, and a stochastic discrete-time polyhedral-slicing setting using Monte Carlo simulation.

\vspace{1em}
 
 \vspace{1em}
 \end{abstract}
 
\section{Introduction}
``The esthetic shapes of mature beach pebbles",  as geologists have remarked,  ``have an irresistible fascination for sensitive mankind"  \cite{stones16}. This fascination dates back at least to Aristotle (\cite{stones30}; see \cite{stones21}), and has often been discussed in the scientific literature (e.g., \cite{stones17}, \cite{stones50}, \cite{stones18}, \cite{stones19},  \cite{stones45}, \cite{stones20},  \cite{stones4}, \cite{stones39}, \cite{stones22},  \cite{stones9},  \cite{stones29},  \cite{stones24},  and \cite{stones43}). 

Various mathematical models for the evolving shapes of 2- and 3-dimensional ``stones", both purely mathematical models on curve-shortening flows and physical models under frictional abrasion, contain hypotheses guaranteeing that the shapes will become spherical in the limit (e.g., \cite{stones12}, \cite{stones27},   \cite{stones7},  \cite{stones3}, and \cite{stones14}). 
Observations of beach stones in nature, however, suggest that the ``esthetically fascinating" shapes of beach stones are almost never spherical. Instead, real beach stones and artificial pebbles from laboratory experiments typically have elongated oval shapes (e.g., see Figures \ref{fig7a} and \ref{fig6}).
Furthermore, in his analysis of these oval shapes (see Figure \ref{fig0}) Black reported that this ``ovoid shape seems to be taken by all sorts of stones, from the soft sandstone to the hard quartzite, and may therefore be independent of mineral composition, or relative hardness of the stone" \cite[p.~122]{stones50}.

 The main goal of this paper is to introduce a simple mathematical equation based on physically intuitive heuristics that may help explain the limiting (non-elliptical) oval shapes of stones wearing down solely by frictional abrasion by waves on a flat sandy beach.  Although very easy to state, this new equation is technically challenging and no closed-form solution is known to the author for most starting stone shapes or distributions of wave energies, even in a 2-d setting. 
 
On the other hand, two different types of numerical approximations of solutions of this equation for various starting shapes indicate promising conformity with the classical experimental and empirical shapes of beach stones found by Lord Rayleigh (son and biographer of Nobelist Lord Rayleigh). One type of numerical solution of the equation models the evolving shapes of various isolated beach stones in a deterministic continuous-time setting using standard  techniques for solving curve-shortening problems, and the other type uses Monte Carlo simulation to approximate typical changes in the stone shape in a discrete-time discrete-state setting. 

\begin{figure}[!htb]
\center\includegraphics[width=0.65\textwidth]{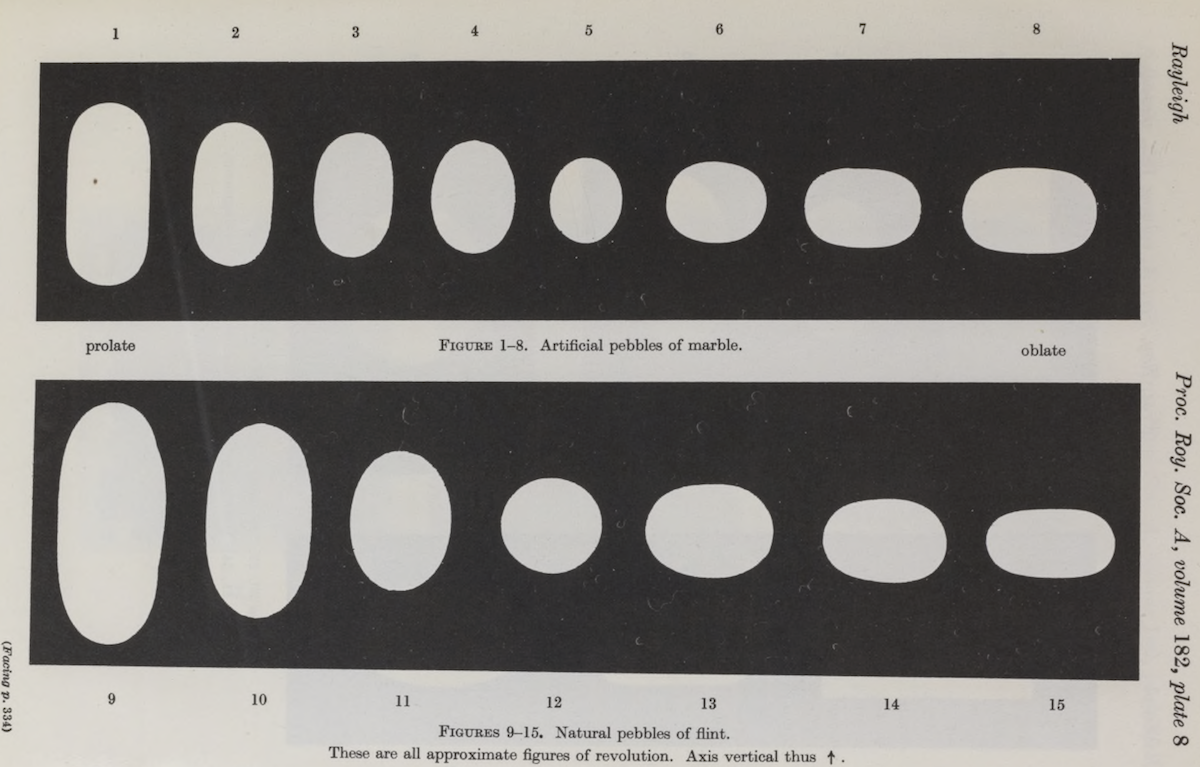}
\caption{Examples of artificial pebbles of marble (top row) abraded in his laboratory, and natural pebbles of flint (bottom row) 
reported by Lord Rayleigh in 1944 \cite{stones10}.}
 \label{fig7a}
\end{figure}

 \begin{figure}[!htb]
\center\includegraphics[width=0.75\textwidth]{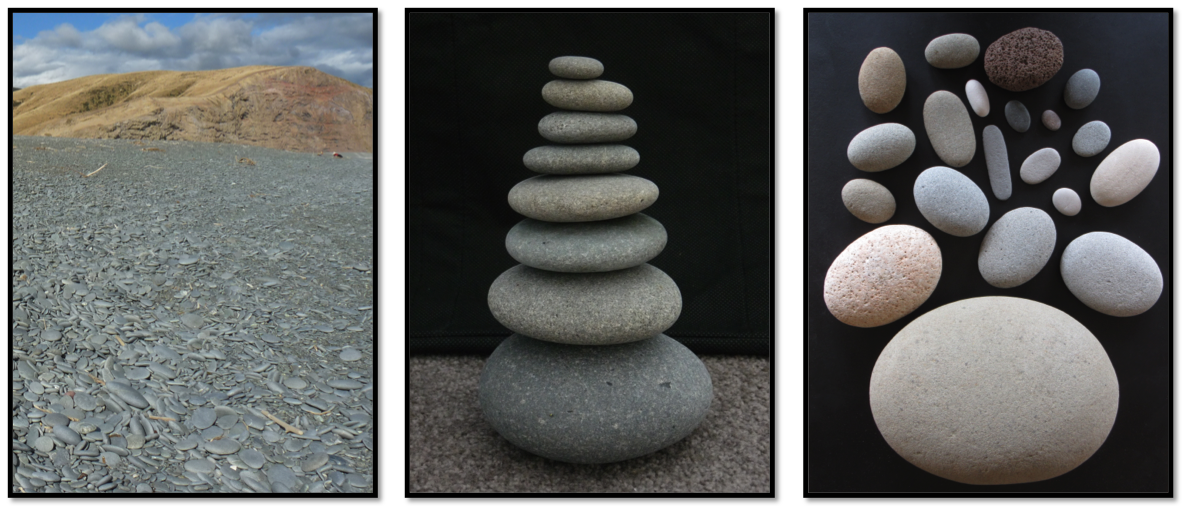}
\caption{Modern beach stones: stones on a beach in the Banks peninsula of New Zealand (left); beach stones collected from a different beach on South Island by A. Berger (center); and beach stones collected by the author on several continents (right; the largest is about 30 cm long, and weighs about 13 kg).}
 \label{fig6}
\end{figure} 

\begin{figure}[!htb]
\center\includegraphics[width=0.85\textwidth]{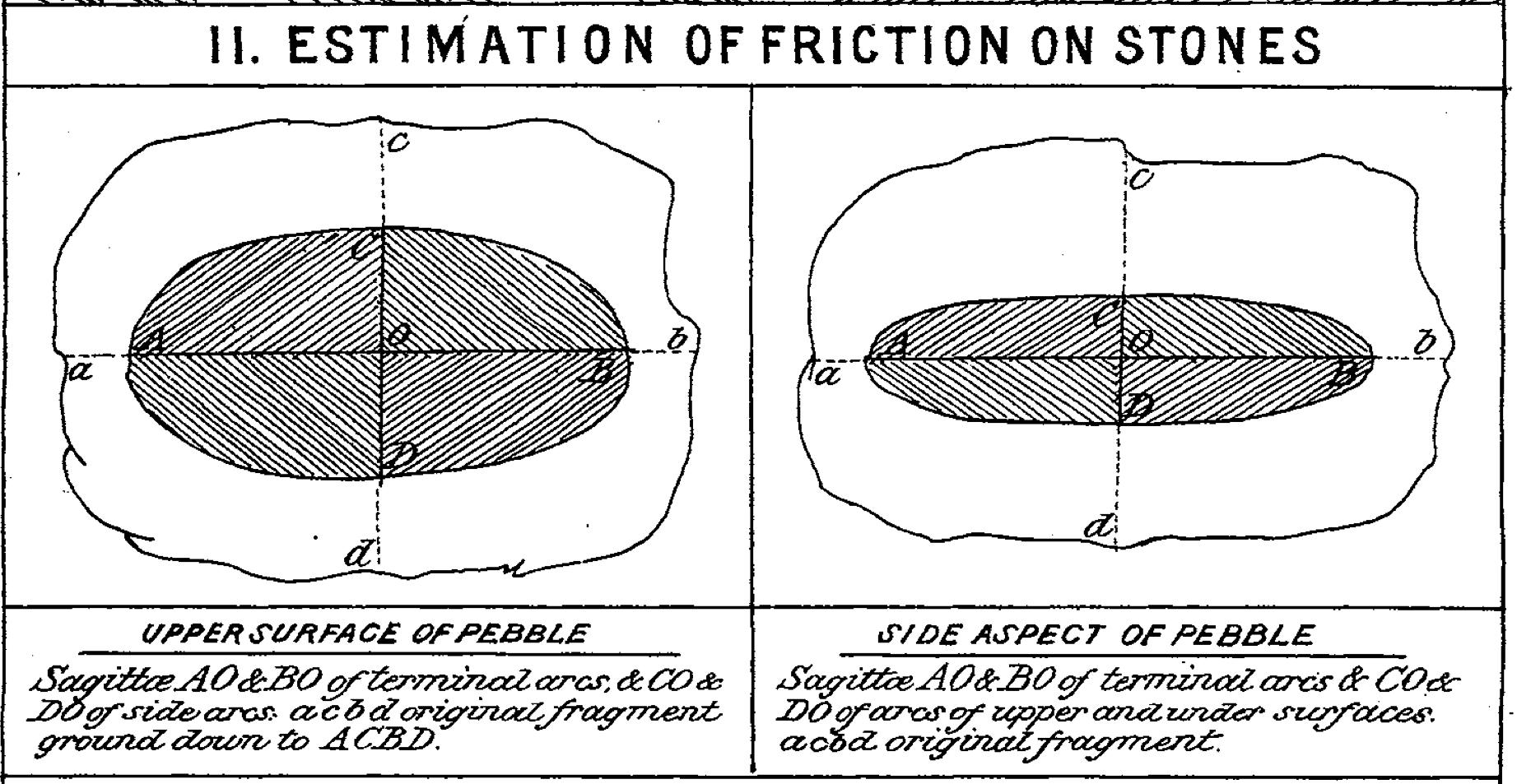}
\caption{Sketch by Black in 1877 illustrating typical dimensions in the top view (left) and side view (right) of a hypothetical worn beach stone  \cite{stones50}.}
 \label{fig0}
\end{figure}

  \begin{remark}\label{rem2}
As observed by Krynine, ``on the seashore the similar pebbles are seen in the same places" \cite{stones21}, and evidence of this is also apparent in Figure \ref{fig6}. Note that stones from the same beach (left and center) appear to have roughly the same shape independent of size - smaller stones do not appear to be becoming spherical or cigar-shaped.  However, shapes of stones from different beaches (right) may vary significantly. In fact, the new model presented below predicts exactly this behavior - that the shapes of stones on the same beach, i.e., subject to the same wave action, tend to evolve toward the same shape, independent of size; see Example \ref{exam1J23} below. 
\end{remark}
 
 This paper is organized as follows: Section \ref{section2} provides an overview of several standard distance-driven and curvature-driven isotropic models of frictional abrasion of stones, with graphical depiction of numerical solutions of each in the 2-d setting; Section \ref{section3} introduces a non-isotropic curvature and contact-likelihood model of frictional abrasion of beach stones; Section \ref{specialCase} introduces a special case of this model, with the new geophysical theory and associated equation; Section \ref{waveDynamics} contains the definition and essential assumptions concerning the underlying wave process;  Section \ref{section5} relates the wave process to the contact-likelihood function;  Section \ref{sectionContactTime} analyzes the shape evolution in the context where both the contact-likelihood function and abrasion are continuous;  Section \ref{sectionMonteCarlo}  analyzes the analogous shape evolution in a discrete time and abrasion setting;  Section \ref{sectionLimitingShapes} addresses limiting shapes of abraded stones under the new non-isotropic frictional abrasion model; and  Section \ref{sectionConclusions} contains a short conclusion.
 
\section{Classical Isotropic Frictional Abrasion Models}\label{section2}

In trying to model the evolving shapes of beach stones, Aristotle conjectured that spherical shapes dominate (see \cite{stones8}).  In support of his theory, he proposed that the inward rate of abrasion in a given direction is an increasing function of the distance from the center of mass of the stone to the tangent plane (the beach) in that direction, the intuition being that the further from the center of mass a point is, the more likely incremental pieces are to be worn off, since the moment arm is larger. 

 Aristotle's model may be formalized as follows: (cf. \cite[equation (1.1)]{stones8}; \cite[equation (1)]{stones45}), 
\begin{equation}
\label{may7eq3}
 \frac{{\partial h}}{{\partial t}} = - f(h),
 \end{equation}
 \begin{align*}
 & \text{where } f \text{ is an increasing function of the distance } h=h(t, u)  \text{ from the center  }  \\
 & \text{ of mass of the stone to the tangent plane in unit direction } u \text{ at time } t.
\end{align*}
 \begin{figure}[!htb]
\center\includegraphics[width=0.71\textwidth]{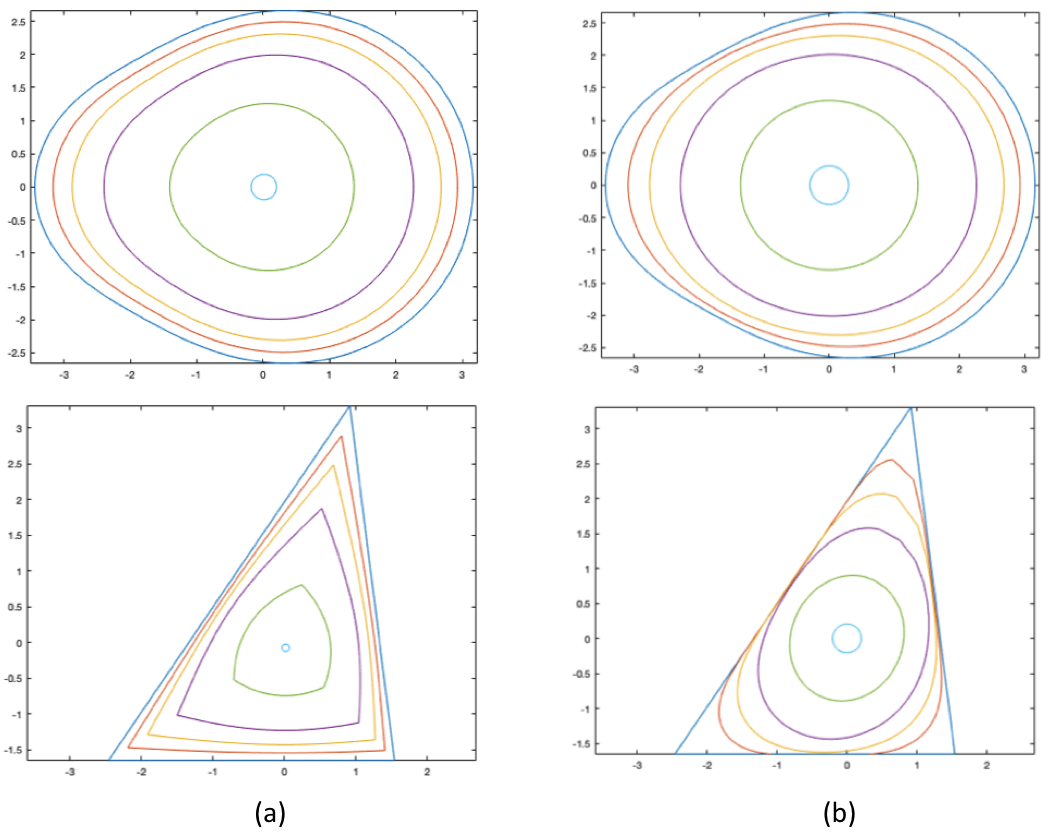}
\caption{Numerical simulations* of an egg-shaped and triangular ``stone" evolving (a) under equation (\ref{may7eq3}) with $f(h)=h^2$, and (b) under  (\ref{eq1}).}
 \label{newFig3}
\end{figure}

*NOTE: The curves in Figure \ref{newFig3} and the other figures below were generated in Matlab using standard semi-implicit curve-shortening techniques (as described in \cite{stones35}) and random number (Monte Carlo) generators. These only yield basic approximations of the evolving shapes, not state-of-the-art renditions. The pseudocode can be found in the Appendix, and the Matlab code for these numerics may be downloaded at \url{https://hill.math.gatech.edu/DOCUMENTS/stonesCode2.zip}.
\vspace{1em}

Figure \ref{newFig3}(a) illustrates numerical solutions of equation (\ref{may7eq3}) in the 2-dimensional setting for triangular and egg-shaped starting shapes for the function $f(h) = h^2$; note that in this case shapes appear to become circular in the limit, which is not the case for some other choices of $f$, such as $f(h) = h$, as will be seen below.
Under this model (\ref{may7eq3}), the further from the center of mass, the faster the stone is eroding. As noted in \cite{stones8}, since the location of the center of gravity is determined by time-dependent integrals, (\ref{may7eq3}) is a non-local (cf.~\cite{stones48}) partial integro-differential equation. 

Modern mathematical models for the evolving shapes of stones often assume, as Aristotle did, that the ablation is normal to the surface of the stone, but unlike Aristotle, they assume that the rate of ablation is proportional to the curvature at the point of contact.  Both Aristotle's and these modern models also assume that the stones are undergoing isotropic abrasion, i.e., the stones are being abraded uniformly from all directions, and each point on the surface of a convex stone is equally likely to be in contact with the abrasive plane. A typical real-life example of isotropic frictional abrasion of a stone is the abrasion of a single stone in a standard rock tumbler.

The assumption that the rate of abrasion at a given point on the surface of the stone is  proportional to the curvature at that point is  analogous to the assumption that equal volumes (areas) are ablated in equal time (see Figure \ref{fig1}). This is physically realistic in that sharp points tend to erode more rapidly than flat regions. Note that under the assumption that the inward rate of abrasion is proportional to the curvature, the stone in Figure \ref{fig1} will erode inward at rates less rapidly from A to C. 

Taking the constant of proportionality to be 1, and using the notation of \cite{stones3}, the basic assumption that the rate of ablation is proportional only to the curvature at the point of contact yields the classical \textit{curvature-driven geometric flow}, the local geometric PDE 
\begin{equation} 
\label{eq1}
\frac{\partial h_0}{\partial t} =  - \kappa 
 \end{equation}
 \begin{align*}
  & \text{where }  h_0 = h_0(t,u) \text{ is the distance from a fixed origin to the } \\
  & \text{ tangent plane in unit direction } u \text{ at time } t \text{ and } \kappa = \kappa(t,u) \text{ is the} \\
  &  \text{ (Gaussian) curvature of the body in unit direction } u \text{ at time } t. \\
 \end{align*}
\begin{figure}[!htb]
\center\includegraphics[width=0.85\textwidth]{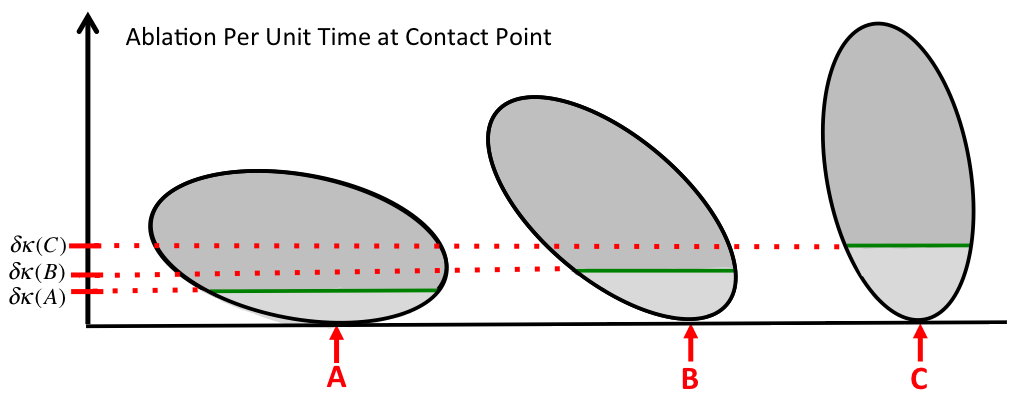}
\caption{In isotropic curvature-driven frictional abrasion models, ablation is assumed inward normal to the surface, at a rate proportional to the curvature at the point of contact. Thus if the curvature $\kappa (A)$ at the point of contact A is half that at C, $\kappa(C)$, the rate at which the surface is being eroded in the normal direction at A is half the rate at C. Note that in Aristotle's distance-driven model (\ref{may7eq3}) the relative rates of erosion here are also increasing from A to C, since the distances from the center of gravity to the point of contact with the abrasive surface are increasing from A to C.} 
 \label{fig1}
\end{figure}
The functions $h$ and $h_0$ in (\ref{may7eq3}) and (\ref{eq1}) are the \textit{support functions} of the stone (simple closed curve or surface) with the origin taken as the center of mass (barycenter) and with the origin fixed, respectively.  As is well-known,  the limiting (renormalized) support function $h_0$ under the curve-shortening flow (\ref{eq1}) is constant for 
essentially all (smooth) convex starting shapes (e.g., \cite{stones1}, \cite{stones2}, and \cite{stones3}). Since support functions uniquely determine convex bodies (e.g., \cite{stones13}), and since spheres are the only convex bodies with constant support functions (with the origin at the center), this implies that the shape of a convex stone eroding under (\ref{eq1}) becomes spherical in the limit; see Figure \ref{newFig3}(b).  (The interested reader is referred to \cite{stones54} for the evolution of shapes under an even broader class of curvature-driven geometric flows.) 

Thus the evolution of shapes of stones under frictional abrasion in distance-driven models such as (\ref{may7eq3}) and curvature-driven models such as (\ref{eq1}) are both isotropic, and both are independent of the shape of the stone away from the point of contact with the beach as well as the underlying wave dynamics.

\section{Non-Isotropic Contact Likelihood}\label{section3}

In a physically realistic model of the evolving shape of a stone undergoing frictional abrasion with a beach, however, both the wave dynamics and the overall shape of the stone play significant roles in the abrasion process. Intuitively, for instance, if the waves are consistently very small the abrasion will be minimal and concentrated on the local stable side of the stone, making it flatter.  Under moderate wave action, however,  beach stones will become more rounded, as will be seen below. The basic assumption here is that whenever there is moving contact of the stone with the beach, friction will occur and the stone will be incrementally abraded at the point of contact.

As for the shape of the stone playing a role, Rayleigh noted that based on his observations in nature and in laboratory experiments, ``this abrasion cannot be merely a function of the local curvature" \cite[p.~207]{stones9}.  Firey similarly observed that the shape of the stone ``surely has a dynamic effect on the tumbling process and so on the distribution of contact directions at time $t$"  \cite[p.~1]{stones3}. 
The distance-driven and curvature-driven models (\ref{may7eq3}) and (\ref{eq1}) do not provide physically realistic frameworks for the evolving shapes of stones undergoing frictional abrasion on a flat beach simply because they \textit{are} isotropic, that is, they assume that abrasion of the stone is equally likely to occur in every direction regardless of the shape of the stone and the dynamics of the wave process. Thus, a more physically realistic model of the evolving shapes of beach stones under frictional abrasion will necessarily be \textit{non-isotropic}. 

In particular, in some models like (\ref{may7eq3}) (with $f(h) = h^\alpha$ for some $\alpha > 1$) and (\ref{eq1}), a spherical stone is in \textit{stable} (\textit{attracting}) \textit{equilibrium}, and any shape close to a sphere will become more spherical. Among real beach stones, however, researchers have reported that ``Pebbles never approach the spherical" \cite[p.~211]{stones29}, ``one will never find stones in spherical form" \cite[p.~1]{stones43}, and ``there is little or no tendency for a pebble of nearly spherical form to get nearer to the sphere" \cite[p.~169]{stones11}. In fact, Landon reported that ``round pebbles become flat" \cite[p.~437]{stones39} and Rayleigh observed ``a tendency to change away from a sphere" \cite[p.~114]{stones9}, i.e., that spheres are in \textit{unstable (repelling) equilibrium}.

To see intuitively how a sphere could be in unstable equilibrium under frictional abrasion alone, consider the thought experiment of the abrasion of a sphere as illustrated in Figure \ref{may23fig1}.
Initially, all points on the surface of the spherical stone are in equilibrium, and the abrasion is isotropic. But as soon as a small area has been ablated at a point on the surface, then that flattened direction is more likely to be in contact with the beach than any other direction, so the abrasion process now has become \textit{non-isotropic}.  That direction of contact with the beach has now entered stable equilibrium, as shown at point B in Figure \ref{may23fig1}.  Moreover, since the center of gravity of the ablated stone has now moved directly away from B, the point A is now also in stable equilibrium, and the stone is more likely to be ablated at A than at any other point except the B side. Thus, if a sphere is subject solely to frictional abrasion with a plane (the beach), the abrasion process will immediately become non-isotropic, and the stone will initially tend to flatten out on two opposite sides.
 \begin{figure}[!htb]
\center\includegraphics[width=0.85\textwidth]{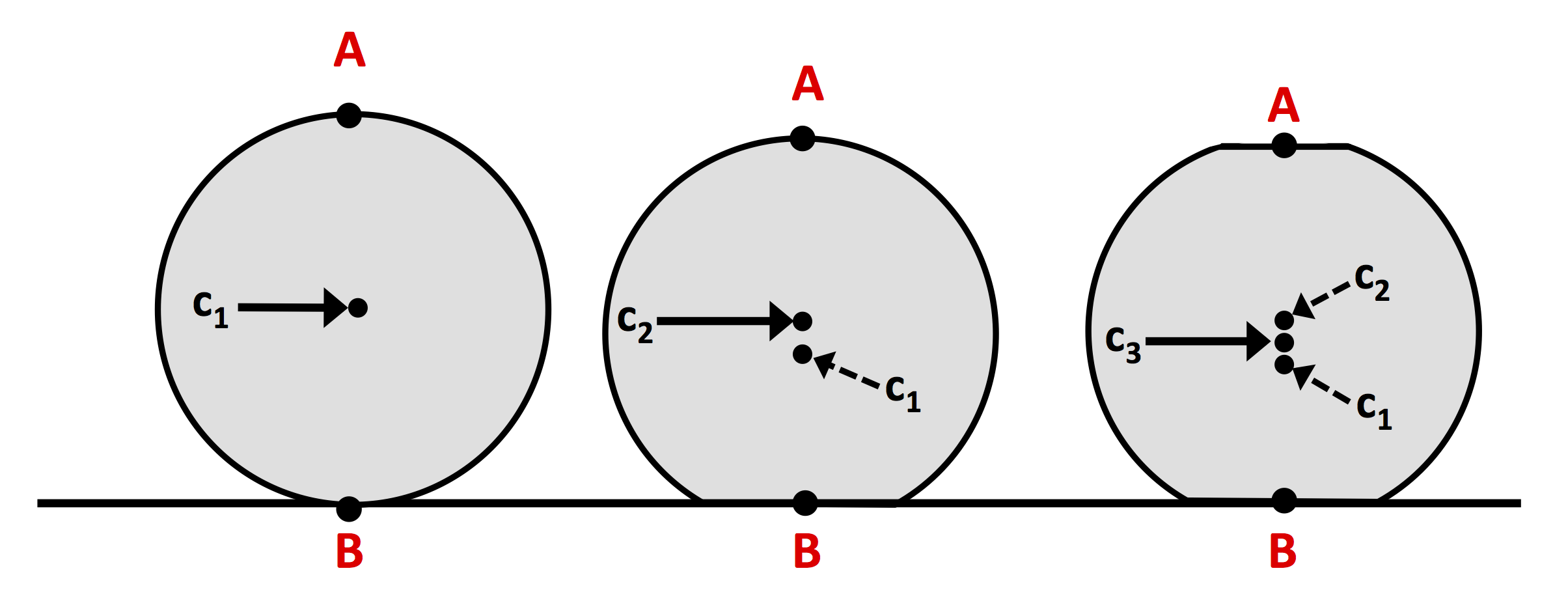}
\caption{In a spherical stone (left) all points on its surface are in unstable equilibrium, with identical curvatures. As  one side is ablated (center) that position now becomes in stable equilibrium, as does the point A diametrically opposite, and the abrasion process becomes non-isotropic; see text.  Hence the most likely directions for the stone to be ablated next are in directions A and B. The centers of gravity of the stones from left to right are at $c_1, c_2, c_3$, respectively.}
 \label{may23fig1}
\end{figure}

As mentioned above, in a stone undergoing frictional abrasion on a beach, not only the shape of the stone, but also the dynamics of the ocean (or lake) waves play a crucial role. If the waves are consistently very small, the stones will tend to rest in one stable position, and the low energy of the waves will cause the stones to grind down to a flat face on that side, much like a standard flat lap polisher is designed to do. The likelihood that other points on the surface of the stone will come into contact with the abrasive beach plane is very small. At the other extreme, if the waves are consistently huge, then it is likely that all exposed surface points of the stone will come into contact with the beach about equally often, i.e., the stone will be undergoing nearly isotropic abrasion as in a rock tumbler, and will become more spherical. 
 
In the non-isotropic model presented below, an isolated beach stone is eroding as it is being tossed about by incoming waves (e.g., the beach may be thought of as a plane of sandpaper set at a slight angle against the incoming waves), and the only process eroding the stone is frictional abrasion with the beach (e.g., no collisional or precipitation factors as in \cite{stones7} or \cite{stones5}). As with the curvature-driven model (\ref{eq1}) above, it is assumed that the \textit{rate of ablation per unit time} at the point of contact with the beach is proportional to its curvature at that point -- that is, sharp points will wear faster than flat regions.  Unlike a stone eroding in a rock tumbler, the likelihood that abrasive contact of a stone with a beach occurs in different directions generally depends on both the shape of the stone and the wave dynamics. That is, in any physically realistic model the ablation process is not isotropic. 

Here it is assumed that the energy required for the frictional abrasion of a beach stone is provided solely by the energy of the incoming waves, a time-dependent random process. Thus the point of contact of the stone with the beach is also a time-varying random variable, and if the inward ablation of a stone at a given point on its surface is an infinitesimal distance $d$ every time that point hits the abrasive surface (beach), then in $n$ hits at that point, the resulting inward abrasion will be $nd$, the \textit{product} of the inward rate and the number of times it is abraded at that point. 

Assuming that abrasion is proportional to curvature, this simple product principle implies that the expected rate of ablation at a point is the product of the curvature there and the average time that point is in contact with the beach, i.e., the likelihood of contact at that point (for details see Section \ref{section5} below).  This suggests the following conceptually natural \textit{curvature and contact-likelihood} equation:

\begin{align} 
\label{eq2}
\frac{{\partial h}}{{\partial t}} =  - \lambda \kappa 
 \end{align}
 \begin{align*}
 &  \text{where } h \text{ is as in (\ref{may7eq3}), }  \kappa  \text{ is as in (\ref{eq1}), and } 
 \lambda = \lambda(t,u) \text{ is the likelihood} \\
 & \text{ of abrasion in unit direction } u \text{ at time } t.
\end{align*}

\begin{figure}[!htb]
\center\includegraphics[width=1.0\textwidth]{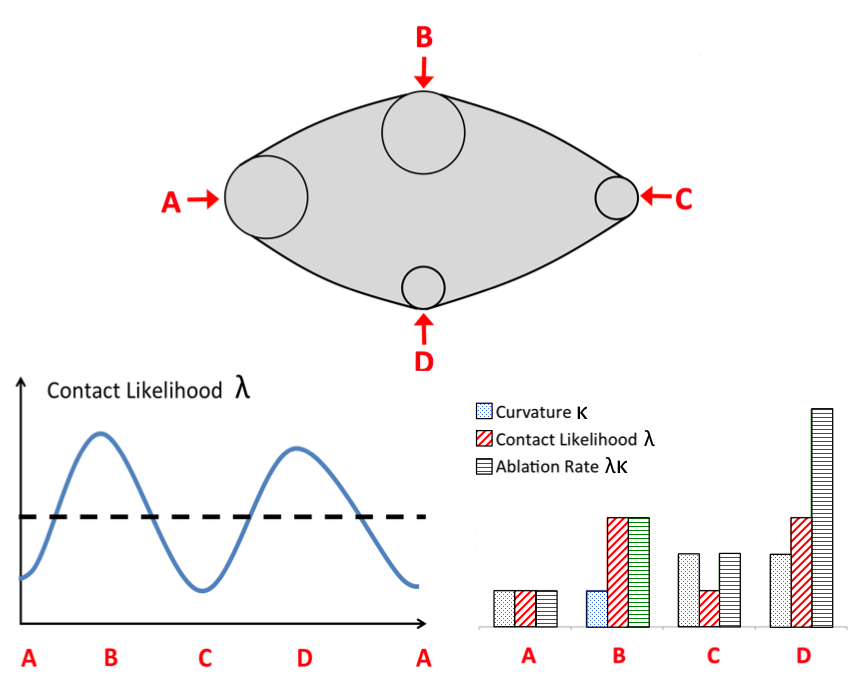}
\caption{The blue solid curve depicts a typical contact likelihood function for the 2-d ``stone" above under moderate wave action, and the dotted line represents the classical isotropic framework where all points on the surface of the stone are equally likely to be in contact with the beach. The graph at right depicts the curvatures, likelihood function values, and the ablation rates predicted by equation (\ref{eq2}) for the points A, B, C, and D. See text.}
 \label{newFig10}
\end{figure}

Figure \ref{newFig10} shows a hypothetical convex 2-d stone with the curvatures $\kappa$ at $A$ and $B$ equal and exactly half the curvatures at $C$ and $D$, as depicted by the osculating circles. In isotropic models of abrasion, where the likelihood function $\lambda$ is constant (dotted black line), the stone is assumed to be in contact with the beach with equal frequency (likelihood) at all points on its surface. In the non-isotropic setting shown (solid blue curve), however, $B$ and $D$ are more often in contact than $A$ or $C$. Under the curvature and contact-likelihood model (\ref{eq2}), the inward ablation rate at $D$ will be larger than the rates at $A$, $B$, and $C$. 

As will be seen in Section \ref{section5} below, $\lambda$ may be viewed as the \textit{local time} of the limiting \textit{occupation measure} (cf.~\cite{stones36}) of the time-dependent random process that reflects which direction the abrasive planar beach will be eroding the stone at that time.
This crucial contact-likelihood function $\lambda$ in (\ref{eq2}) may in general be very complicated since it can  depend on both the shape of the stone away from the point of contact, and on the dynamics of the underlying wave process.  

\section{An Illustrative Model} \label{specialCase}

In cases where the wave process and stone satisfy standard regularity conditions, as will be seen next, the contact-likelihood function $\lambda$ may sometimes be approximated by a very simple function of the distance from the center of mass of the stone. A concrete example of this is $\lambda = h^{-\alpha}$ for some $\alpha \geq 1$, in which case the basic non-isotropic principle (\ref{eq2}) becomes the non-linear partial integro-differential equation
\begin{equation}\label{eq3}
\frac{{\partial h}}{{\partial t}} =  - \frac{\kappa }{h^\alpha} \
\end{equation}
\begin{align*}
 & \text{where } h = h(t,u)  \text{ is as in (\ref{may7eq3}), }  \kappa  \text{ is as in (\ref{eq1}), and } \alpha \geq 1.
\end{align*}
\noindent
\textbf{Thought Experiment for Equation (\ref{eq3})}. \\
\indent
Consider a single fist-sized non-spherical convex stone (such as one of those in Figures \ref{fig7a} or \ref{fig6}) that is eroding by friction as it is rolled about by waves on a flat sandy beach. Assume that when a wave comes in, it rolls/lifts the stone to a point of contact with the beach such that the potential energy of the stone in that position is proportional to the wave crest height - larger waves lift the center of mass of the stone higher. 

Next note that if the expected number of waves needed until the stone comes into contact with the beach at point B on the surface of the stone is twice the expected number of waves needed until point A comes into contact, then over time point B will be in contact with the beach half as often as point A. That is, the likelihood of the stone being in contact with the beach at a given point on the surface of the stone is proportional to the \textit{reciprocal of the expected waiting time} to hit that point. 

Assuming the random wave crest heights follow a Pareto distribution (a standard assumption; see below), the expected number of waves until a crest of height at least $h$ occurs is proportional to $h^\alpha$ , so the contact-likelihood is proportional to its reciprocal $1/h^\alpha$. Since the instantaneous rate of ablation at a point of contact is assumed to be proportional to the curvature   $\kappa$ at that point (again, sharper points erode faster), taking the constant of proportionality to be $1$  gives $\kappa/h^\alpha$, which yields the shape evolution equation (\ref{eq3}). (More formally, see Section \ref{sectionContactTime} below.)\\

The novelty of this model (\ref{eq3}) is that it considers the \textit{product} (rather than the sum) of curvature and distance driven terms and it derives the distance-driven term from the Pareto distribution of waves. 

As suggested in this thought experiment, the new PDE model (\ref{eq3}) is only intended to model the shape evolution of intermediate macroscopic-sized beach stones, not huge boulders or grain-sized particles. Thus this model is clearly not valid in the limit, where the beach stone eventually becomes another grain of sand. Also, (\ref{eq3}) does not address the changing size of the abrading stone, which in some river rock models has been shown to decrease exponentially in time (e.g., see \cite{stones45} and \cite{stones55}).

In model (\ref{eq3}), the different roles of the three essential rate-of-abrasion factors -- curvature at point of contact, global shape of the stone, and wave dynamics -- are readily distinguishable in the three variables $\kappa$, $h$, and $\alpha$. The variable $\kappa$ reflects the curvature at the point of contact, $h$ reflects the global shape of the stone via its evolving center of mass, and $\alpha$ reflects the intensity of the wave process (in fact, in the interpretation in Example \ref{exam1J23} below, $\alpha$ is an explicit decreasing function of the expected (mean) value of the wave crests).  For example, increasing the curvature at the point of contact affects neither the center of mass nor the wave dynamics, changing the center of mass affects neither the curvature at the point of contact nor the wave dynamics, and changing the wave dynamics affects neither the center of mass nor the curvature of the stone.  

 \begin{figure}[!htb]
\center\includegraphics[width=0.75\textwidth]{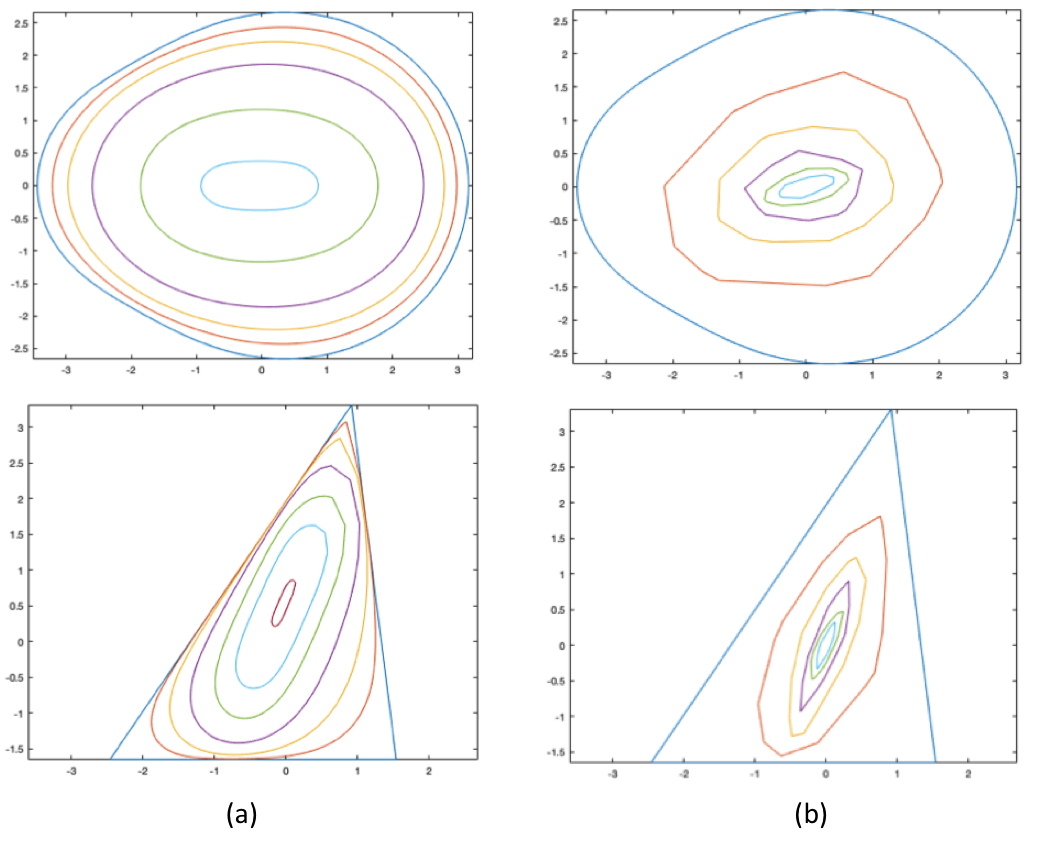}
\caption{Numerical approximations of the shapes of two intermediate-sized 2-d ``stones" evolving (a) under equation (\ref{eq3}) with 
$\alpha = 3$, and (b) under a discrete chipping analog of (\ref{eq3}), as will be discussed in Section  \ref{sectionMonteCarlo}.  These illustrations depict evolving shapes, not limiting or attracting shapes.}
 \label{newFig11}
\end{figure} 
Figure \ref{newFig11}(a) illustrates the same initial 2-d stones shown in Figure \ref{newFig3} evolving under the non-isotropic curvature and contact-likelihood equation (\ref{eq3}) with $\alpha=3$;  note the similarity of these oval shapes with the real beach stones in Figures \ref{fig7a} and \ref{fig6}.
As will be seen in Section \ref{sectionMonteCarlo} below, this same prototypical equation model  (\ref{eq3}) also appears to be robust in both 2-d and 3-d discrete-time discrete-state ``stochastic-slicing" models of the evolution of shapes of beach stones, as illustrated in Figures \ref{newFig11}(b) and \ref{allStones1}.

\begin{remark}\label{rem3b}
Replacing the curvature term $\kappa$ in equations (\ref{eq1}), (\ref{eq2}) and (\ref{eq3}) by its positive part $\kappa ^{+}  = \max\{\kappa, 0\}$ allows those abrasion models to be extended to non-convex initial shapes, and this is left to the interested reader (see Figure \ref{photo2}).
\end{remark}

\begin{figure}[!htb]
\center\includegraphics[width=0.85\textwidth]{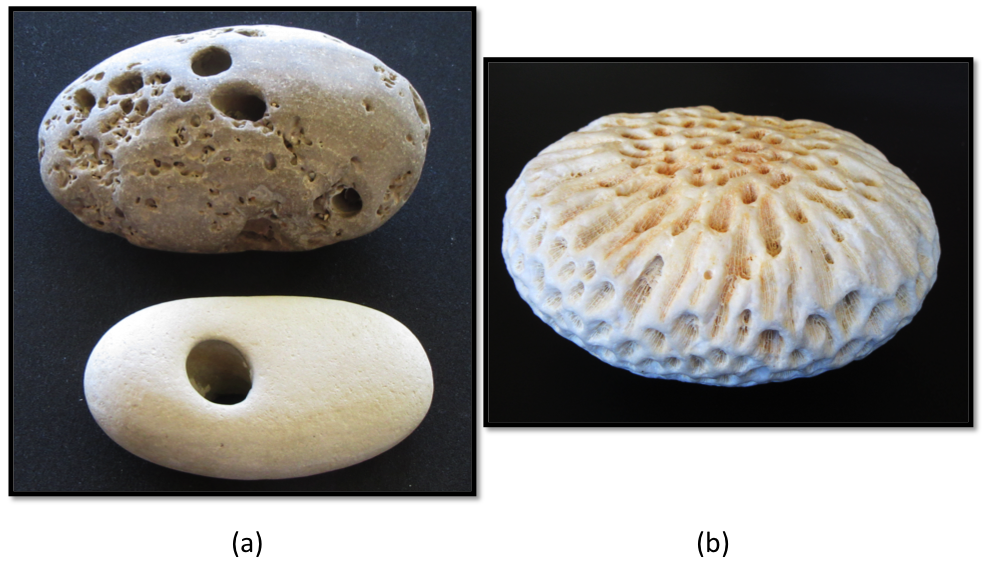}
\caption{Three isolated beach stones collected by the author illustrate the apparent prevailing oval shapes of beach stones even when the stone is not homogeneous.  The holes in the two stones in (a) were made by a boring clam \textit{triodana crocea} in the face of an underwater stationary rock wall or boulder at Monta{\~n}a de Oro State Park in California.  These oval-shaped ``holey" stones were formed when portions of those rocks with the clam holes broke off and were worn down by frictional abrasion with the beach. The coral stone in (b) is from a beach cave in Negril, Jamaica.
}
 \label{photo2}
\end{figure}

\begin{remark}\label{rem3}
Several other theories have also been proposed to model the evolution of non-spherical stone shapes. For example, \cite{stones5} 
 shows the existence of non-spherical limits for a process that involves a model with a combination of collisional abrasion, frictional abrasion and isotropic surface growth, and \cite{stones47} is a PDE model for purely collisional abrasion that also predicts non-spherical limit shapes. In fact, the evolving shape of the egg-shaped stone in Figure \ref{newFig11}(a) is strikingly similar to Figure 4 in \cite{stones44}, which is based on a model leading to the formation of elliptical stones by both grinding and rolling abrasion. For an overview of existing models of the geometry of abrasion, including their history, new extensions, and the mathematical relationship between various models, the reader is referred to \cite{stones45}. 
 \end{remark} 
 
 \section{Wave Dynamics}\label{waveDynamics}
 
The first step in formalizing the fundamental role played by the waves in this non-isotropic curvature and contact-likelihood frictional abrasion model (\ref{eq2}) is to define formally what is meant by a wave process. Real-life ocean or lake waves are random processes whose components (velocity, direction, height, temperature, etc.) vary continuously in time.  For the purposes of the elementary abrasion model introduced here, it will be assumed that the critical component of the wave is its energy, and for simplicity, that this is proportional to its height.

From a realistic standpoint, it is also assumed that wave energies are \textit{bounded}, i.e., not infinitely large, and that the wave crests (local maxima) are \textit{isolated}, that is, no finite time interval contains an infinite number of crests.  These simple notions lead to the following working definition.

\begin{dfn}
\label{dfn1}
A \textit{wave process $W$} is a bounded, continuous, real-valued stochastic process with isolated local maxima on an underlying probability space $(\Omega, \mathcal{F}, P)$, i.e., $W:\Omega \times \mathbb{R}^+ \rightarrow \mathbb{R}$ is such that:
\begin{equation}
\label{wp1}
\begin{split}
& \mbox{for each } \omega \in \Omega, W(\omega,\cdot): \mathbb{R}^+ \rightarrow \mathbb{R} \mbox{ is bounded and} \\
 &\mbox{ continuous  with isolated local maxima}; and \\
\end{split}
\end{equation}
\begin{equation}
\label{wp2}
\mbox { for each } t \ge 0, W(\cdot, t) \mbox{ is a random variable.}
\end{equation}
\end{dfn}

A standard assumption in oceanography, (e.g., see \cite{stones42},  \cite{stones38},  \cite{stones40}, and  \cite{stones41}),
is that the wave crests have a Pareto distribution. The next example describes a wave process with this property, and, as seen above, this Pareto distribution plays a key role in the basic heuristics underlying the physical intuition for equation (\ref{eq3}).

\begin{exam}\label{4junExample2}
\vspace{1em}
\noindent
$W(t) = X_{\left \lfloor{t} \right \rfloor} sin (2\pi t)$,
where $\left \lfloor{t} \right \rfloor = \max\{n:n \leq t\}$, and $X_1, X_2, $ $X_3, \ldots$ are {i.i.d.} Pareto random variables 
with {c.d.f.} $P(X_j \leq x) = 1 - ({x_0}/x) ^2$ for all $x \geq x_0 > 0$; see Figure  \ref{waveBrownian}. 
In contrast to the wave models in \cite{stones43} and \cite{stones44}, here $W$ is  \textit{not periodic}, as real waves are not,  since  the sequence of wave crest heights of $W$ are {i.i.d.} Pareto random variables. \end{exam} 

Note that unlike Brownian motion, which is also a continuous-time continuous-state stochastic process, a wave process is in general not Markov for the simple reason that the current instantaneous state of the process alone may not indicate whether the wave is rising or falling. 

\begin{figure}[!htb]
\center\includegraphics[width=0.75\textwidth]{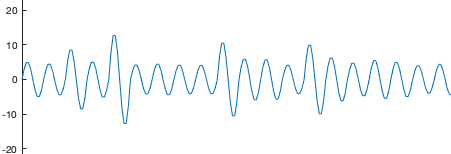}
\caption{A sample path of a stochastic wave process with Pareto distribution as in Example \ref{4junExample2}. Note that the process is not periodic, which plays a crucial role in the theory presented here.}
 \label{waveBrownian}
\end{figure}

Another key assumption about the wave process necessary for a physically realistic frictional abrasion process to follow the curvature and contact-likelihood model (\ref{eq2}) is that \textit{the long-term behavior of the wave process is in equilibrium} (\textit{steady state}). To put this in context, recall that for a continuous function $X:\mathbb{R}^+ \rightarrow \mathbb{R}$, the \textit{occupation measure} (or \textit{occupation time of $X$ up to time $s$})  is the function $T_{X,s} = T_s$ defined by 
\begin{equation*}
T_s (B) = m(\{0 \leq t \leq s:X(t) \in B\}) = \int_0^s I_B (X(t)) dt  \quad  \mbox { for all Borel } B \subset \mathbb{R}, 
\end{equation*}
where $m$ denotes Lebesgue measure on $\mathbb{R}^1$, and $I_B$ is the indicator function of $B$. 

\subsection*{Wave Steady-State Assumption}

In addition to the wave continuity assumption (\ref{wp1}), it is assumed that the wave process  $\{W(\omega, t): t \ge 0\}$
 has a limiting average occupation measure $\mu_W$, i.e., there is a Borel probability measure $\mu_W$  on $\mathbb{R}$  satisfying 
 
\begin{equation} \label{assumption2}
{\mu_W}(B) = \lim_{s \to \infty} \frac{1}{s} m(\{0 \le t \le s : W(\cdot,t) \in B\}) 
 \mbox{ a.s. for all Borel }
B \subset \mathbb{R}.
\end{equation}

 Note that assumption (\ref{assumption2}) is essentially a strong law of large numbers, and implies for instance that $W$ is not going off to infinity, or forever oscillating on average between several different values.

\begin{exam}\label{4junExample1}
Suppose $W$ is a wave process with Pareto distribution as in Example \ref{4junExample2}.  Then the maximum heights of the wave intervals $\{X_j \sin(2 \pi t):t \in [j, j+1); j \ge 1\}$ are $X_1, X_2, \ldots,$ respectively, which by assumption are i.i.d. Pareto with $P(X_j > x) = (x_0/x)^2 $ for all $x \ge x_0$.  Thus by the Glivenko-Cantelli Theorem, the equilibrium limiting distribution of the wave crests of $W$ has this same Pareto distribution. 
\end{exam}

\section{Contact-Likelihood Function}\label{section5}

The next step in relating the underlying wave process $W$ to equation (\ref{eq2}) is to describe the relationship of $W$ to the contact-likelihood function $\lambda$, which involves the direction of the point of abrasion on the stone as a function of the underlying time-dependent stochastic wave process $W$. For ease of exposition and grapical illustration, in this section ``stones" will be depicted in a 2-d setting.

Even for homogeneous and strictly convex stones, the role played by the contact-likelihood function $\lambda$ distinguishes the dynamics of the evolution of shape given by the non-isotopic model (\ref{eq2}) from isotropic distance-driven models like (\ref{may7eq3}) and from isotropic curvature-driven models such as (\ref{eq1}), as is illustrated in Figure \ref{newFig10}.

A (2-dimensional) \textit{stone} is a compact convex set $K \subset \mathbb{R}^2$ with non-empty interior $int(K)$.  Let $c = c_K \in int(K)$ denote the center of mass (barycenter) of $K$, and let $S^1$ denote the unit ball $S^1=\{(x,y) \in \mathbb{R}^2:x^2+y^2=1\}$.

\begin{dfn}\label{dfn2}
An \textit{oriented stone} $\gamma$ is an embedding $\gamma:S^1 \rightarrow \mathbb{R}^+$ with the origin taken as the barycenter of the convex hull of the graph of $\gamma$. Let $\mathbb{S}$ denote the set of all oriented stones.
\end{dfn}

The point of abrasion of a stone with the beach as a result of an incoming wave depends not only on the size and shape of the stone, but also on the wave energy and the orientation of the stone with the beach when the wave hits.  Note that the same wave may act on different orientations of the same stone to bring it into contact with the beach at different points on its surface.  

\begin{dfn}\label{dfn3}
An \textit{abrasion direction function} $D$  is a continuous function $D: \mathbb{S} \times \mathbb{R} \rightarrow S^1$.
\end{dfn}
 The value $D(\gamma,z)$ specifies the unit direction of the abrasion plane (the beach) resulting from a wave with energy $z \in \mathbb{R}$ acting on the oriented stone $\gamma$.  In other words, $D(\gamma, z)$ specifies which direction of $\gamma$ will be ``down" after $\gamma$ is hit by a wave with energy $z$. 
Figure \ref{newFig16}(a) illustrates typical values of $u_1$ and $u_2$ of the abrasion function $D$ of the oriented stone $\gamma$ after impact by two waves with different wave energies $z_1$ and $z_2$, respectively, resulting in two different points of contact with the beach, $u_1$ and $u_2$, at distances $h(u_1)$ and $h(u_2)$ from the center of mass.  

\begin{figure}[!htb]
\center\includegraphics[width=0.85\textwidth]{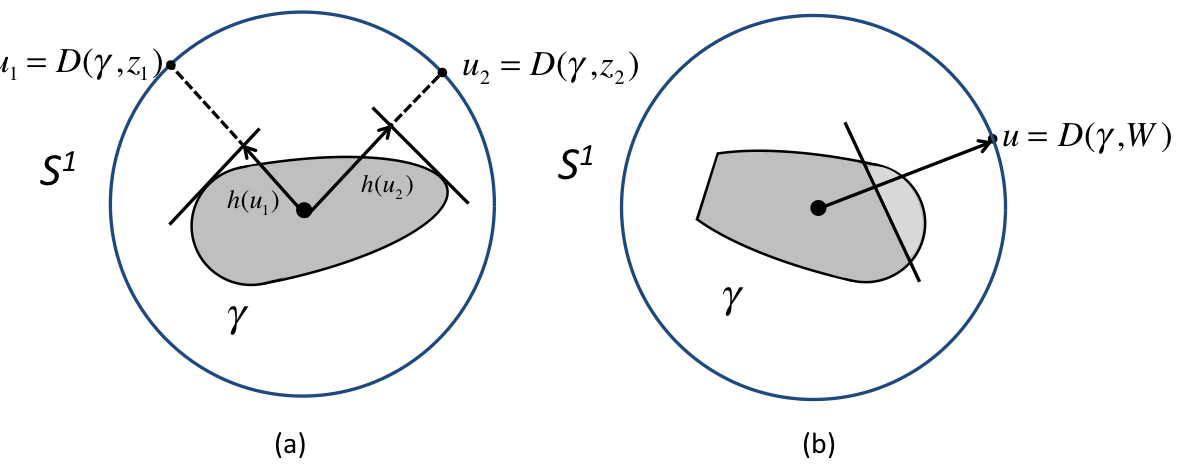}
\caption{In (a) an oriented stone $\gamma$  hit by two different waves with parameters $z_1$ and $z_2$, respectively, results in two different directions of contact with the abrasive plane (beach), $u_1 = D(\gamma, z_1)$ and $u_2 = D(\gamma , z_1)$. Analogously, (b) illustrates the same process in a discrete slicing framework as will be discussed in Section \ref{sectionMonteCarlo} below.} 
 \label{newFig16}
\end{figure}

Recall that $m$ denotes Lebesgue measure.

\begin{prop}
\label{may7prop3}
Given an oriented stone $\gamma$, a wave process $W$, and an abrasion direction function $D$, the function 
$\Lambda = \Lambda(\gamma, W, D):(\Omega, \mathcal{F}) \rightarrow [0,1]$ given by
\begin{equation} \label{assumption4}
\Lambda(B) = \lim_{s \to \infty} \frac{1}{s}m(\{0 \le t \le s : D(\gamma,W(\cdot,t)) \in B \})
\mbox { for all Borel } B \subset S^1
\end{equation}
almost surely defines a Borel probability measure on $S^1$.
\end{prop}

\begin{proof}
Fix a Borel set $B$  in $S^1$.  Recall by (\ref{wp2}) that for all $t \ge 0$, $W(\cdot,t)$ is a random variable.
By (\ref{wp1}) and Definition \ref{dfn3},   $D(\gamma, \cdot)$ is continuous, and hence Borel measurable, so there exists a Borel set $\hat{B}$  in $\mathbb{R}$ such that 

\begin{equation}
\label{may7eq2}
D(\gamma, W (\cdot, t))  \in B) \iff W(\cdot, t) \in \hat{B} \, \mbox { for all } t \ge 0.
\end{equation}
By the wave steady-state assumption (\ref{assumption2}), the limit in (\ref{assumption4}) exists and equals $\mu_W (B)$ a.s., so since $\mu_W$ is a probability measure, $0 \le \Lambda(B) \le 1$ a.s.  The demonstration that $\Lambda$ is a.s. a measure is routine.
\end{proof}
\vspace{1em}

The probability measure $\Lambda$ in Proposition \ref{may7prop3} is the \textit{occupation measure} 
(cf.~\cite{stones36}) of the steady-state likelihood (average time) that the oriented stone $\gamma$ is in contact with the abrasive plane in various directions, assuming that the rate of abrasion is negligible.  For example, if $I \in S^1$ is an interval of unit directions, then $\Lambda(I)$ is the probability that the oriented stone $\gamma$ is in contact with the beach in direction $u$ for some $u \in I$.  

If $\Lambda$ is absolutely continuous (with respect to Lebesgue measure on $S^1$), then $\lambda$, the Radon-Nikodym derivative of $\Lambda$ with respect to the uniform distribution on $S^1$, is the \textit{local time} (cf.~\cite{stones37}) of the stochastic process $D(\gamma, W)$.  That is, 
$\lambda = d\Lambda/dm$ is the density function of the distribution of the occupation measure.  In some instances, as will be seen in the next section, $\lambda$ may be approximated by a simple function of $\gamma$, in particular, of the support function $h$ of $\gamma$.  

\section{A Model With Continuous Contact-Likelihood and Abrasion}\label{sectionContactTime}

Here, the energy required to produce frictional abrasion of a stone on the beach is assumed to come only from the waves, which lift and slide the stone against the beach (recall that in this simple model, collisional abrasion with other stones is assumed negligible).  To lift the stone in Figure \ref{fig3} to abrasion position (c) requires more energy than to lift it to position (b), and (b) requires more energy than (a). Thus the expected likelihood (or frequency that) the stone is in position (c) is less than that in (b), and (b) less than (a). This means that for these three points of contact, the value of the contact-likelihood function $\lambda$ is decreasing from (a) to (c);  the actual numerical values of $\lambda$ at these points of course also depend on the external wave process.
\begin{figure}[!htb]
\center\includegraphics[width=0.8\textwidth]{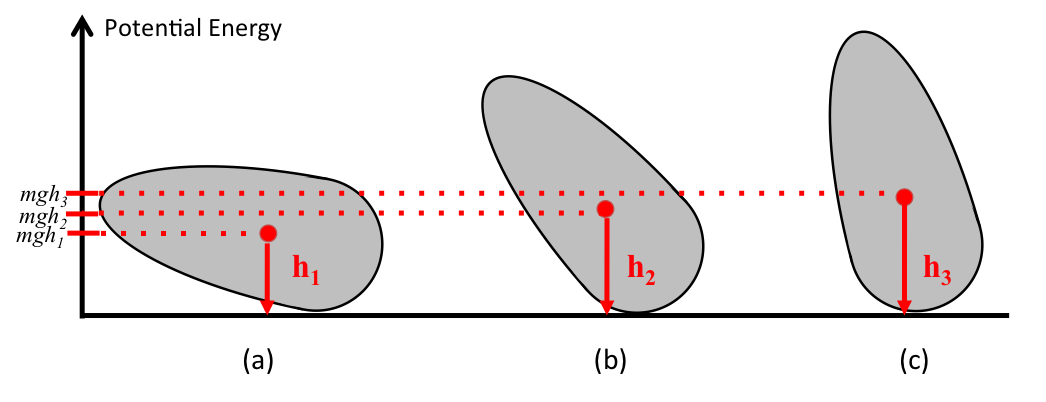}
\caption{The distances $h$ from the center of gravity of the stone in the direction of the normal to the tangent contact plane are proportional to the potential energies of the stone in that position, and hence proportional to the wave energy necessary to lift the stone to that position.}
 \label{fig3}
\end{figure}

Consider a model of the curvature and contact-time ablation equation (\ref{eq2}) where the ablation process is assumed to be continuous in time and space, i.e., a curve-shortening process (cf.~\cite{stones35}).  As before, the incoming wave crest of  $W$ lifts the stone to a position determined by the wave parameters (e.g., kinetic energy of the crest), where its surface is ablated incrementally.  

Fix $t>0$, and suppose that the oriented stone $\gamma = \gamma(t)$ is smooth and strictly convex, i.e., the non-empty interior of $\gamma$ is strictly convex with smooth $(C^\infty)$ boundary.  Since the support function $h$ is continuous, there exist $0 < h_{min} < h_{max} < \infty$ so that on $\gamma$ 
\begin{equation}\label{eq2J23}
\mbox{range}(h) = [h_{min}, h_{max}] \subset \mathbb{R}^+.
\end{equation}
Let $\Lambda$ denote the occupation measure of the likelihood function of the abrasion direction process as in Proposition \ref{may7prop3}, and let $X_\Lambda$ denote a random variable with values in the unit sphere and with distribution $\Lambda$, i.e., for all intervals of unit directions $I$, 
$P(X_{\Lambda} \in I) = \Lambda(I)$ represents the likelihood that $\gamma$'s direction of contact with the planar beach at time $t$ is in $I$. Assuming that $W$ and $D$ are continuous ((\ref{wp1}) and Definition \ref{dfn3}), it is routine to check that since $\gamma$ is strictly convex, the random direction $X_{\Lambda}$ is absolutely continuous.  Thus $X_{\Lambda}$ has a (Borel) density function $\lambda : S^1 \rightarrow \mathbb{R}^+$ satisfying $P(X_{\Lambda} \in I) = \int_I \lambda(u) du$ for all intervals $I \subset S^1$.

Let $Y_\Lambda$ denote the random variable $Y_\Lambda = mghX_\Lambda$, where $m$ is the mass (e.g., volume, or area in the 2-d setting) of $\gamma$ and $g$ is the force of gravity. Thus $Y_\Lambda$ represents the potential energy of $\gamma$ when $X_\Lambda$ is the direction of contact of the 
stone $\gamma$ with the abrasive plane, i.e., when $X_\Lambda$ is the ``down" direction at time $t$.  Then (\ref{eq2J23}) implies that 
\begin{equation}\label{eq3J23}
\mbox{range}(Y_{\Lambda}) = [mgh_{min}, mgh_{max}] \subset \mathbb{R}^+.
\end{equation}

Assuming that the wave crest energies (relative maxima) are converted into potential energy of the stone in the corresponding ``down" positions (see Figure \ref{fig3}), this implies that the distribution of $Y_\Lambda$, given that $Y_\Lambda$ is in $[mgh_{min}, mgh_{max}]$, is the same as the distribution of the successive wave crests of $W$ (see Figure \ref{waveBrownian}) given that they are in  $[mgh_{min}, mgh_{max}]$.

\vspace{1em}

Ignoring secondary effects such as multiple rolls of the stone, this yields an informal physical explanation for equation (\ref{eq3}) with $\alpha = 3$, as is seen in the next example.

\begin{exam}\label{exam1J23}
Suppose that $\gamma$ is smooth and strictly convex and that $W$ is a wave process as in Example \ref{4junExample2}, with the $\{X_j\}$ {i.i.d.} Pareto random variables satisfying $P(X_j > x) = (x_0/x)^2$ for all $x \geq x_0$ for some $x_0 > 0$.  Then the sequence $X_1, X_2, \ldots$ represents the values of the successive crests (relative maxima) of $W$, 
i.e., $X_j = \max\{W(\cdot, t) : t \in [j, j+1)\}$ (see Figure \ref{waveBrownian}).

This implies that for all $x_0<x_1<x_2$, the conditional distribution of each $X_j$ given that $X_j$ has values in $[x_1, x_2]$ is an absolutely continuous random variable with density proportional to $1/x^3$ for $x \in [x_1, x_2]$, i.e., there is a $c > 0$ so that 
\begin{equation}\label{eq4J23}
P(X_j  \in I  \mid  X_j \in [x_1, x_2]) = c \int_I  \frac{1}{x^3} dx  \mbox{ for all } I = (a_1, a_2) \subset [x_1, x_2].
\end{equation}
Letting $Y_j$ denote the maximum potential energy of the stone $\gamma$ during time period $[j, j+1)$, then $Y_j = mgh(X_j)$ (see Figure \ref{fig3}).  Again assuming that the wave energy at its crests are converted into potential energy of the stone (see Figure \ref{3blocks}), (\ref{eq4J23}) implies that $Y_j$ is also absolutely continuous with density proportional to $1/h^3$ for $h \in [h_{min}, h_{max}]$, so  (\ref{eq2}) yields (\ref{eq3}) with $\alpha=3$. Note that as the stone gets smaller, the factor $1/h^3$ remains unchanged, but is applied to new values of $h_{min}$ and $h_{max}$. This suggests that stones of different sizes on the same beach, i.e., subject to the same (Pareto) wave process, will abrade toward the same (renormalized) shapes; see Figure \ref{fig6} and Remark \ref{rem2}.  
\end{exam}

\begin{remark}\label{rem4}

Note that the model in (\ref{eq3}) is not valid for extremely small values, e.g. when the size of the beach stone is below the Pareto threshold $x_0$ of the wave. Intuitively, when the beach stone becomes extremely small, it is comparable to one of the grains of sand that make up the beach, and is subject to different dynamics such as collisional abrasion and fracturing.
\end{remark}

\section{A Discrete Contact-likelihood and Abrasion Model}\label{sectionMonteCarlo}

In actual physical frictional abrasion, of course, the evolution of the shape of a stone is not continuous in time, since the ablated portions occur in discrete packets of atoms or molecules. For isotropic frictional abrasion, this has been studied in \cite{stones47}, \cite{stones46}, and \cite{stones26}, where analysis of the evolution of the rounding of stones uses Monte Carlo simulation and a ``stochastic chipping" process.  The goal of this section is to present an analogous stochastic discrete-time analog of the evolution of a stone's shape under the basic isotropic curvature and contact-likelihood equation (\ref{eq3}), where again discrete portions of the stone are removed at discrete steps (see Figure \ref{newFig16}(b)), but now where the effects of both the global shape of the stone (via $h$) and the wave dynamics (via $\alpha$) are also taken into account.

To see how a contact-likelihood function $\lambda$ may be discrete and explicitly calculated (or approximated), consider the 2-dimensional rectangular ``stone" in Figure \ref{3blocks}.  Without loss of generality, $x_1 < x_2$ and $m = 2/g$, so the potential energy of the stone in position (a) is $x_1$ and the potential energy in position (c) is $x_2$.

\begin{figure}[!htb]
\center\includegraphics[width=0.85\textwidth]{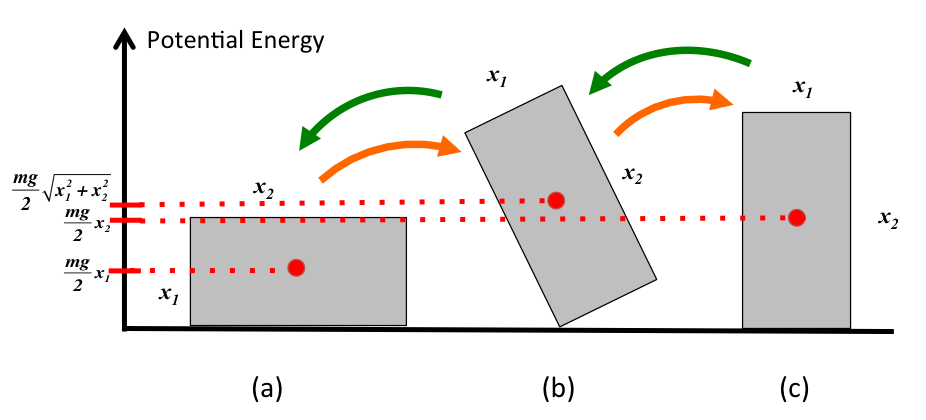}
\caption{A rectangular stone has stable positions of equilibrium shown in (a) and (c); more energy is required to move the stone from position (a) to (c) than to move it from (c) to (a).}
 \label{3blocks}
\end{figure}

Let $W$ be a Pareto wave process as in Example \ref{4junExample2}; recall Figure \ref{waveBrownian} for a sample path. Then the crests (maximum wave heights) of $W$ are the {i.i.d.} random variables $\{X_j; j \in \mathbb{N} \}$. Let $\bar{F}$ denote the complementary cumulative distribution function of $X_1$, i.e., $\bar{F} = P(X_1 > x)$ for all $x \ge 0$.

Suppose first that the stone is on a longer $x_2$-side (Figure \ref{3blocks}(a)) at time $j \in \mathbb{N}$. Then it flips onto an $x_1$-side (Figure \ref{3blocks}(c)) during the time interval $[j, j+1)$ if and only if the value of $X_j$ is greater than the energy required to lift the stone from position (a) to position (b), i.e., is enough to increase the potential energy of the stone from $x_1$ to more than $\sqrt{{x_1}^2 + {x_2}^2}$.

Since the $\{X_j\}$ are {i.i.d.}, this implies (ignoring multiple flips) that the number of waves until a flip occurs from (a) to (c) is a geometric random variable $N_1$ with parameter $\bar{F} (\sqrt{{x_1}^2 + {x_2}^2} - x_1):= p_1$, so the expected value of $N_1$ is $E(N_1) = 1/p_1$.  Similarly, the expected number of waves $E(N_2)$ until a flip occurs from a shorter $x_1$-side (Figure \ref{3blocks}(c)) to an $x_2$-side is $1/p_2$, where $p_2 = \bar{F} (\sqrt{{x_1}^2 + {x_2}^2} - x_2) > p_1$.

Thus by the strong law of large numbers, the limiting frequency of time that the stone is on side $x_1$ is less than the relative frequency of time the stone is on side $x_2$, since 
\begin{equation*}
\frac{E(N_2)}{E(N_1) + E(N_2)} = \frac{1/p_2}{1/p_1 + 1/p_2} 
= \frac{p_1}{p_1 + p_2} <  \frac{p_2}{p_1 + p_2} = \frac{E(N_1)}{E(N_1) + E(N_2)}.
\end{equation*}

\vspace{1em}
\begin{exam}\label{jun13ex1}
Suppose the 2-dimensional stone is as in Figure \ref{3blocks} with $x_1 = 6$ and $x_2 = 8$, and the relative maxima (crests) of the wave process $W$ are as in Example \ref{4junExample2}. Then 
\begin{equation*}
p_1 = \bar{F}(4) = \frac{c}{4^2} > \frac{c}{2^2} = \bar{F}(2) = p_2,
\end{equation*}
so the likelihood that the stone is on a short side ($x_1$ or its opposite side) is $\frac{{2^2}/c}{{2^2}/c + {4^2}/c} = 0.2$ and the likelihood the stone is on a long side ($x_2$ or its opposite) is $0.8$.  This implies that in terms of the \textit{oriented} stone as in Figure \ref{3blocks}(a), the contact likelihood function $\lambda$ at time $t$ satisfies $\lambda(t, (1,0)) = \lambda(t, (-1,0)) = 0.1,  \lambda(t, (0,1)) = \lambda(t, (0,-1)) = 0.4$, and
$\lambda(t, u) = 0$ for $u \notin \{(1,0), (-1,0), (0, 1), (0, -1)\}$.
\end{exam}

In this setting, as illustrated in Figure \ref{newFig16}(b), an oriented stone $\gamma$ is hit by a wave $W$ resulting in the unit direction of contact $u=D(\gamma, W)$ of $\gamma$ with the abrasive plane, at which time a small fixed fraction $\delta$ of the volume of the stone is ground off in that direction.  (Recall as illustrated in Figure \ref{fig1} that removing a fixed fraction of the stone in a given direction is analogous to removing a portion proportional to its curvature there.) The evolving stones in this discrete stochastic framework are eventually random convex polygons (polyhedra), for which almost every point on the surface has curvature zero.  Thus this assumption that fixed proportions are removed, rather than portions proportional to curvature, seems physically intuitive.

Figure \ref{newFig11}(b) illustrates the results of a Monte Carlo simulation of this stochastic-slicing process evolving under the discrete-time analog of equation (\ref{eq2}) in the special case (\ref{eq3}) with $\alpha = 3$ for the same two initial 2-d stone shapes as in Figure \ref{newFig3}. Here, the direction of ablation is again selected at random, not uniformly (isotropically), but inversely proportional to the cube of the distance in that direction from the center of mass to the tangent plane (line). Note the apparent similarity of the evolving oval shapes in both the continuous and discrete settings, as seen in Figure \ref{newFig11}. 

An analogous Monte Carlo simulation of this same frictional abrasion process is illustrated in the 3-d setting in Figure \ref{allStones1} where  two initial shapes, a smooth convex egg-shaped body and a non-regular tetrahedron, are undergoing a discrete-time analog of the same basic non-isotropic curvature and contact-likelihood model (\ref{eq3}) with $\alpha = 3$.  Similar to the analysis in \cite{stones26} where a discrete-time {\textit stochastic chipping} model of the isotropic curvature-driven equation (\ref{eq1}) was used to study the rate at which initial 3-d shapes converge toward spheres, the evolving body here repeatedly has sections of a fixed proportion $\delta$ of the volume removed at each step by a planar cut, in a random direction, normal to the support function in that direction. In this case, however, in sharp contrast to that in \cite{stones26}, the abrasion is non-isotropic with the likelihood of abrasion in a given direction inversely proportional to the cube of the distance from the center of mass of the stone to the supporting plane in that direction (cf.\ Example \ref{exam1J23}).

\begin{figure}[!htb]
\center\includegraphics[width=0.85\textwidth]{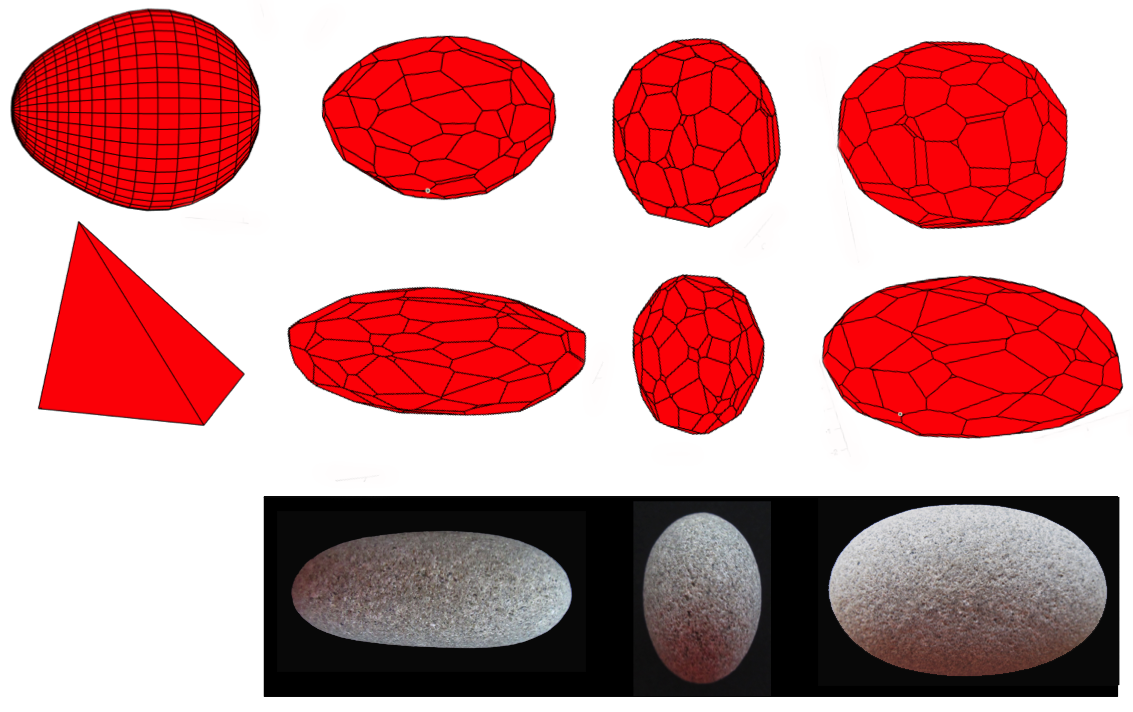}
\caption{Monte Carlo simulations, in the 3-d setting, analogous to the 2-d results illustrated in Figure \ref{newFig11}(b), where fixed proportions of the volume are sliced off in random directions, with the directions chosen inversely proportional to the cube of the distance from the center of mass, i.e. a discrete analog of (\ref{eq3}) with $\alpha=3$. For comparison, the corresponding ``side", ``end", and ``top" views of one of the natural beach stones in Figure \ref{fig6} (right) are shown at bottom.} 
 \label{allStones1}
\end{figure}

\section{Limiting Shapes in a Continuous Abrasion Setting}\label{sectionLimitingShapes}

Recall that, as the numerical approximations in Figure \ref{newFig3}(b) illustrated in the 2-d setting, the limiting shape of stones under curvature-only ablation (\ref{eq1}) is spherical, and when normalized, is the unit sphere.  Similarly, if the contact-likelihood function $\lambda$ in (\ref{eq2}) is constant, then the process is isotropic and the author conjectures that the limiting shape will also be spherical.  

For non-isotropic (non-constant) contact likelihood functions $\lambda$, however, the limiting shape depends on $\lambda$, and this shape may sometimes be determined or approximated as follows.
First, it is routine to check that the re-normalized shapes will remain the same if and only if $h=h(u, t)$ satisfies
\begin{equation}
\label{eqShape1}
\frac{{\partial h}}{{\partial t}} = - c h 
\end{equation} 
for some  $c> 0$. 
Equating the term $\partial{h}/\partial{t}$ in equation (\ref{eqShape1}) with the same term in (\ref{eq2}) yields the \textit{shape equation}
\begin{equation}
\label{eqShape2}
\kappa = c\frac{h}{\lambda}.
\end{equation} 

Suppose that the underlying wave crests have a Pareto distribution with $\alpha > 1$, and that the ablation process results solely from the conversion of the energy of the wave process $W$ into the potential energy of the stone, by lifting it to the position where abrasion will occur.  
Then, as seen in Example \ref{exam1J23}, $\lambda$ is proportional to  $h^{-\alpha}$.  With (\ref{eqShape2})  this yields the \textit{limiting shape equation} 
\begin{equation}\label{eqShape3}
\kappa = ch^{\alpha+1}.
\end{equation}

Numerical solutions of (\ref{eqShape3}) for $c=1$ and $\alpha=2.5,3,4$ are shown in Figure \ref{fig5}.  Note that flatter ovals correspond to Pareto waves with smaller means (i.e., with lighter tails), that is, as physical intuition suggests, more powerful waves produce more spherical limiting shapes.

Since $\kappa = (h + h'')^{-1}$, and since $h$ is the distance to the center of mass, note that (\ref{eqShape3}) is a non-local ordinary differential equation. 
 
\begin{figure}[!htb]  
\center\includegraphics[width=0.85\textwidth]{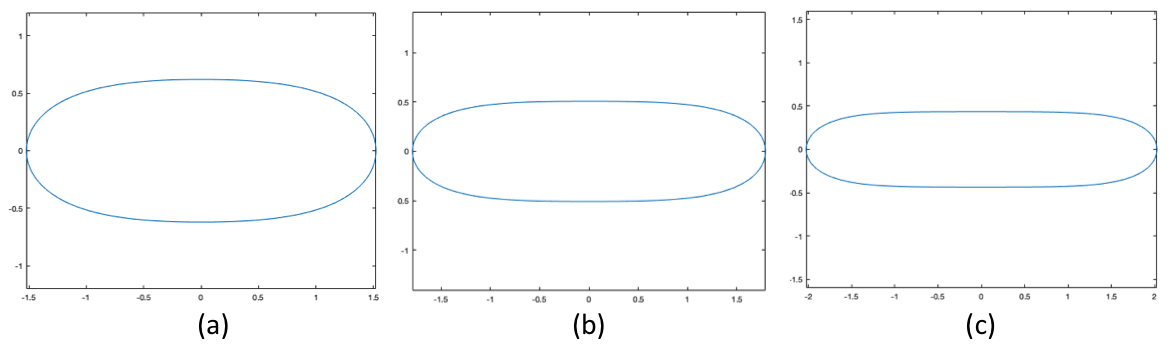}
\caption{Plots of numerical solutions of the limiting shape equation (\ref{eqShape3}) for $\alpha = 2.5, 3.0$ and $4.0$, respectively.}
 \label{fig5}
\end{figure}

Note that the oval shapes in Figure \ref{fig5} appear very similar to the non-elliptical ovals found by Rayleigh shown in Figure \ref{fig7a} in his empirical data in both natural specimens of beach stones and in his laboratory experiments. There is also a close resemblance of these same shapes to those in Fig. 4 of \cite{stones5}, which studies the evolving shapes of carbonate particles that are both growing from chemical precipitation and eroding from physical abrasion.

The equations for these ovals, the solutions of (\ref{eqShape3}), are not known to the author.  
Moreover, as Rayleigh noted, ``the principal section of the pebble lies outside the ellipse drawn to the same axes, and I have not so far found any exception to this rule among artificial pebbles shaped by mutual attrition, or among natural pebbles'' \cite{stones10}.

\begin{figure}[!htb]
\center\includegraphics[width=0.85\textwidth]{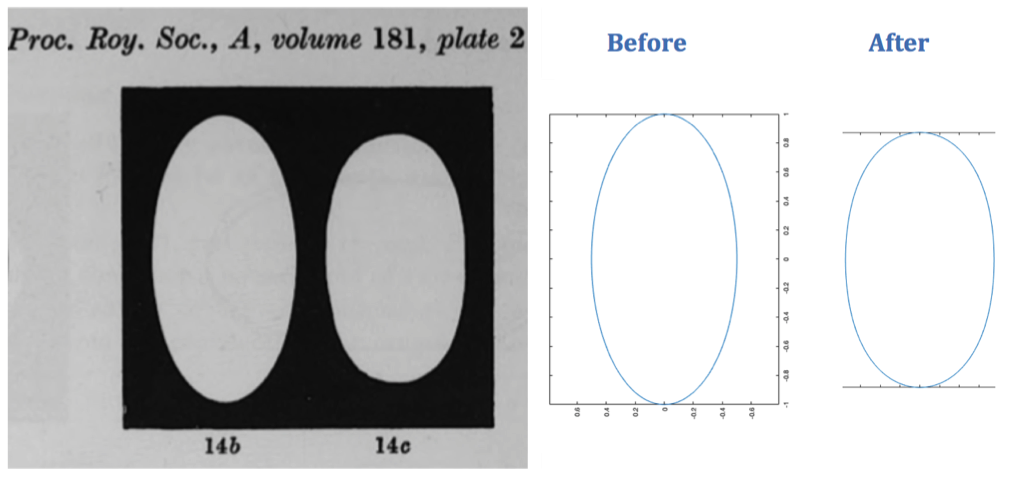}
\caption{Actual before and after shapes (left) of a stone worn by frictional abrasion in a laboratory, as recorded by Lord Rayleigh in 1942.
 Rayleigh specifically noted that the limiting shapes are not ellipses, and demonstrated this experimentally by starting with a stone with an elliptical shape (14b), which after ablation assumed the non-elliptical shape shown in (14c) \cite{stones9}. The graphics on the right illustrate how closely a numerical solution to equation (\ref{eq3}) with $\alpha=2.2$ approximates his findings in a 2-d setting.}
 \label{4junRayleigh}
\end{figure}

More concretely, Figure \ref{4junRayleigh} illustrates the limiting shapes predicted by the prototypical model in equation (\ref{eq3}) with the empirical laboratory data reported by Rayleigh \cite{stones9}.  The stone (14b) in Figure \ref{4junRayleigh} is the near-elliptical actual stone he subjected to frictional abrasion, and to its right (14c) is the same stone after ablation.  In the two curves on the right in Figure \ref{4junRayleigh}, the one on the left is an exact ellipse with minor axis 0.5 and major axis 1.0 centered at the origin, and to its right is the evolved shape after curve shortening via the curvature and contact-likelihood equation (\ref{eq3}) with $\alpha = 2.2$.  Note the striking resemblance of the experimental results with the model presented in the simple equation (\ref{eq2}) with $\lambda = h^{- \alpha}$.   

Also note that the chipping process depicted in Figure \ref{allStones1}, even when renormalized at each step, does not have a limiting shape since at each successive iteration a fixed proportion of the volume is chipped off. In this scenario the author conjectures that there is a \textit{limiting distribution} of (convex polyhedral) shapes, and those illustrated in Figure \ref{allStones1} are representative cases.

\section{Conclusions}\label{sectionConclusions}

The non-isotropic model of the evolution of the shapes of beach stones introduced here is meant as a starting point to include the effects of both the global (non-local) shape of the stone and the wave dynamics into the process. The main equations are simple to state, but  as non-local partial integro-differential equations, they are difficult to solve exactly, and no solutions are known to the author. Numerical approximations in the continuous-time continuous-state framework using standard curve-shortening algorithms, and in the discrete-time discrete-state framework using Monte Carlo simulation, both indicate remarkably good agreement with the shapes of both natural and artificial stones undergoing frictional abrasion on a flat plane. 

\section*{Acknowledgements}
The author is grateful to Professors Pieter Allaart, Arno Berger, G{\'a}bor Domokos, Lester Dubins, Ron Fox, Ryan Hynd, Kent Morrison, and Sergei Tabachnikov for many helpful comments, to Professor Donald Priour for access to his 3-d ``stochastic-chipping" code, to an anonymous referee for many corrections and helpful suggestions,  and especially to John Zhang for his excellent work on the Monte Carlo simulations and curve-shortening numerics presented here, and for many helpful ideas, suggestions, and questions. 


\clearpage
%
%
\section*{Appendix - Pseudocode for Figures}

\noindent
\begin{tabbing}
\textbf{Figure \ref{newFig3}(a): } 
\=  Set $S$ = one of the two shown 2-d stones (outer curves) $S$ in Figure \ref{newFig3}.\\
\> START\\
\>Calculate the center of mass $c_s$ of $S$.\\
\>Compute incremental new shape $S_1$ using a stable explicit scheme (no tangential motion) \\
\> \indent for curve-shortening of $S$ under $dh/dt = -h^2$, where $h$ is the support function of $S$ \\
\> \indent with $c_s$ as origin. \\
\>Set $S = S_1$, return to START.
\end{tabbing}
\noindent
Code: \textit{Aristotle.m}
\vspace{1.5em}

\noindent
\begin{tabbing}
\textbf{Figure \ref{newFig3}(b): } 
\= Fix the origin O, and center all the stones so that center of mass is O. \\
\>Set $S$ = one of the two 2-d stones (outer curves) $S$ in Figure \ref{newFig3}. \\
\>START \\
\>Compute incremental new shape $S_1$ using a semi-implicit finite difference scheme \\
\> \indent for curve-shortening of $S$ under $dh/dt = -k$, where $h$ is the support function \\
\> \indent of $S$ with origin at O. [Note: this scheme takes a $C^1$, closed, embedded plane \\
\> \indent curve and deforms it for the life of the flow.] \\
\>Set $S = S_1$, return to START.
\end{tabbing}
\noindent
Code: \textit{New\_CSF\_Semi\_Implicit\_6.m}
\vspace{1.5em}

\noindent
\begin{tabbing}
\textbf{Figure \ref{newFig11}(a):} 
\=Same as Figure \ref{newFig3}(a) except using $dh/dt = -k/h^3.$
\end{tabbing}
\noindent
Code: \textit{New\_CSF\_Semi\_Implicit\_6.m}
\vspace{1.5em}

\noindent
\begin{tabbing}
\textbf{Figure \ref{newFig11}(b):} 
\=Set $S$ = one of the two 2-d stones (outer curves) $S$ in Figure \ref{newFig11}. \\
\>START \\
\>Calculate the center of mass $c_s$ and the area  $A_s$ of $S$. \\
\>Shift the shape so that its center is at the origin. \\
\>Generate a random angle  $\theta$ uniformly in $[0, 2\pi]$, and let  $\theta_j = {\theta + 2\pi j}/8, j = 1, \ldots, 8$.\\
\>Choose an angle $\Theta$  at random among the ${\theta_j}$  inversely proportional to $h^{3}(\theta_j)$ , where \\
 \> \indent $h$ is the support function of $S$ with origin at (0, 0). \\
\> Compute distance $d$ of the line perpendicular to $\Theta$, in the direction of $\Theta$ \\
\> \indent from $c_s$, so that it cuts off 0.01$A_s$. \\ 
\> Compute new shape $S_1$ after this cut. \\
\>Set $S = S_1$, return to START.
\end{tabbing}
\noindent
Code: \textit{DiscretizedStones.m}
\vspace{1.5em}

\noindent
\begin{tabbing}
\textbf{Figure \ref{waveBrownian}:} 
\=Fix a period $P = 2\pi$ \\
\>Generate $N = 20$ independent Pareto values $X_1, \ldots,X_{20}$ with mean 2. \\
\>Generate standard sin wave values. \\
\>For the $j^{th}$ period, multiply by $X_j$. 
\end{tabbing}
\noindent
Code: \textit{may2waves.m}
\vspace{1.5em}

\noindent
\begin{tabbing}
\textbf{Figure \ref{allStones1}:} 
\= 1. First, produce the initial shape. We do this by denoting all vertices of the polygon, \\
\> \indent and then creating a mesh-grid out of those vertices (library does this by finding \\
\> \indent the convex shape with vertices, faces). For the eggshape and ellipsoid, we pass \\
\> \indent in the spherical coordinates and allow the function to create the mesh-grid. \\
\> \indent For the trapezoid, we use a library function and immediately pass in the \\
\> \indent mesh-grid values. \\
\> 2.  Calculate the original volume \\
\>3. Initialize xyz-coordinates of 12 equally spaced points, chosen as vertices of the \\
\> \indent isocahedron. \\
\>4. START \\
\>5. Center the shape \\
\>6. Create a random 3D rotation of the 12 vertices, using the yaw, pitch, and roll \\ 
\> \indent rotation matrices \\
\>7. Calculate the distance h to the polygon surface in these 12 directions, and \\
\> \indent choose a direction with probability proportional to $1/h^3$. \\
\>8. If deterministic, move cuts along this direction incrementally, stopping when a  \\
\> \indent perpendicular plane cuts away delta*volume\_of\_shape on the previous iteration. \\ 
\> \indent If random, choose a distance uniformly. Inspect the perpendicular cut made at this \\
\> \indent distance along the chosen direction. Accept this cut with probability exponentially \\
\> \indent decreasing in volume cut away, proportioned so that the average ratio of volume \\
\> \indent cut is delta. \\
\>9. Determine the new volume, return to START.
\end{tabbing}
\noindent
Code: \textit{PolygonSlicing3D.m}
\vspace{1.5em}

\noindent
\begin{tabbing}
\textbf{Figure \ref{fig5}: } 
\=Set $S$ = an ellipse with minor axis 0.7, major axis 1, centered at the origin. \\
\>START \\
\>Compute incremental new shape $S_1$ using a semi-implicit finite difference scheme \\
 \> \indent    for curve-shortening of $S$ (with no tangential motion) under $dh/dt = -k/h^{expNum}$, \\
\> \indent     where $expNum$ is a variable input to the program: 2.5 for (a), 3 for (b), and 4 for (c). \\
\>Resize the shape to retain the same area. \\
\>STOP IF all coordinates of the current shape differ from all coordinates of the previous    \\     
\> \indent by less than $10^{-6}$, i.e., the limiting shape of this equation has been reached. \\
\>Set $S = S_1$, return to START. 
\end{tabbing}
\noindent
Code: \textit{Numerical\_Solve\_Curve\_2.m}
\vspace{1.5em}

\noindent
\begin{tabbing}
\textbf{Figure \ref{4junRayleigh}: }
\= Same as Figure \ref{fig5}, except $S$ = an ellipse with minor axis 0.5, major axis 1, \\
\> \indent and expNum = 2.2.
\end{tabbing}
\noindent
Code: \textit{Numerical\_Solve\_Curve\_2.m}

%
%

\end{document}